\documentclass[pra,reprint,doublecolumn,superscriptaddress, nobalancelastpage,dvipsnames]{revtex4-2}
\usepackage[dvipsnames]{xcolor}

\usepackage{tikz}
\usepackage{quantikz}
\usepackage{tikz-feynman}
\usetikzlibrary{shapes.misc, positioning, backgrounds}
\usepackage{tikz-cd}
\usepackage{pgfplots}
\usepgfplotslibrary{groupplots}
\usepgfplotslibrary{dateplot}
\usepackage[sort&compress]{natbib}
\usepackage{verbatim}
\usepackage{float}
\usepackage{sidecap}
\sidecaptionvpos{figure}{c}
\usepackage{lipsum}
\usepackage{placeins}
\usepackage{capt-of}
\usepackage{setspace}
\usepackage{dcolumn}
\usepackage{bm}
\usepackage{amssymb}
\usepackage{amsmath}
\usepackage{amsthm}
\usepackage{mathtools}
\usepackage{booktabs}
\usepackage{lipsum}  
\usepackage{physics}
\usepackage{multirow}
\usepackage[a4paper]{geometry}
\newgeometry{top=1.25cm,right=1cm,left=1.25cm,bottom=1.25cm}
\usepackage[colorlinks=true, urlcolor=blue, pdfborder={0 0 0}]{hyperref}
\hypersetup{linkcolor=blue,citecolor=BrickRed}

\DeclareMathOperator*{\argmax}{argmax}
\newcommand{\x}{\boldsymbol{x}}

\definecolor{darkgreen}{rgb}{0.0,0.60,0.30} 
\definecolor{newred}{rgb}{1.00,0.70,0.70} 

\newtheorem{prop}{Proposition}
\newtheorem{corol}{Corollary}
\draft%
\begin{document}

\title{\LARGE{\normalfont{Provably Trainable Rotationally Equivariant Quantum Machine Learning\\}}}

\author{Maxwell T. West} \email{westm2@student.unimelb.edu.au}  \affiliation{School of Physics, The University of Melbourne, Parkville, 3010, VIC, Australia}
\author{Jamie Heredge}   \affiliation{School of Physics, The University of Melbourne, Parkville, 3010, VIC, Australia}
\author{Martin Sevior} \affiliation{School of Physics, The University of Melbourne, Parkville, 3010, VIC, Australia}
\author{Muhammad Usman} \email{musman@unimelb.edu.au}  \affiliation{School of Physics, The University of Melbourne, Parkville, 3010, VIC, Australia}
\affiliation{Data61, CSIRO, Clayton, 3168, VIC, Australia}

\maketitle%

\onecolumngrid%

\noindent%
\textcolor{black}{\normalsize{\textbf{Exploiting the power of quantum computation to realise superior machine learning algorithms 
has been a major research focus of recent years, but the prospects of quantum machine learning (QML) remain dampened by considerable technical challenges.
A particularly significant issue is that generic QML models suffer from so-called barren plateaus in their training landscapes -- 
large regions where cost function gradients vanish exponentially in the number of qubits employed, 
rendering large models effectively untrainable.
A leading strategy for combating this effect is to build problem-specific models which take into account the symmetries of their data in order to focus on a
smaller, relevant subset of Hilbert space. In this work, we introduce a family of rotationally equivariant QML models built upon the quantum Fourier transform,
and leverage recent insights from the Lie-algebraic study of QML models to prove that (a subset of) our models do not exhibit barren plateaus.
In addition to our analytical results
we numerically test our rotationally equivariant models on a dataset of simulated scanning tunnelling 
microscope images of phosphorus impurities in silicon, where rotational symmetry naturally arises,
and find that they 
dramatically outperform their generic counterparts in practice.
}}}
\\ \\ \\
\twocolumngrid%
\noindent%
\section{Introduction}

Stimulated by steady progress in quantum computing hardware~\cite{bravyi2022future}, the possibility of using parameterised quantum circuits to carry out machine 
learning tasks has been thoroughly investigated in recent years as a near-term
application of noisy intermediate scale quantum computers~\cite{biamonte2017quantum,beer2020training, havlivcek2019supervised, schuld2019quantum,qcnn,west2023towards,tsang2023hybrid,schuld2021supervised,liu2021rigorous,huang2022quantum}.
Among the most important discoveries in quantum machine learning (QML) has been the existence of barren plateaus in the training landscapes of variational quantum 
models~\cite{mcclean2018barren,wang2021noise,holmes2022connecting,cerezo2021cost},
reminiscent of the vanishing gradients which plagued early neural networks~\cite{bengio1994learning}.
Barren plateaus have come to be 
understood to be linked to the expressibility of the ansatz being optimised~\cite{holmes2022connecting,larocca2022diagnosing,ragone2023unified,fontana2023adjoint},
with generic, highly expressible models rendered effectively untrainable as the number of qubits increases. 
This suggests that the architectures of variational QML models should be tweaked on a per-problem basis, with domain specific knowledge employed to construct 
ansatze with inductive biases compatible with the optimal solution.
A natural approach along these lines is to build models that explicitly respect the symmetries of their data, so-called \textit{geometric quantum machine learning} 
(GQML)~\cite{meyer2022exploiting,ragone2022representation,nguyen2022theory,schatzki2022theoretical,sauvage2022building,skolik2022equivariant,larocca2022group,zheng2022benchmarking,west2023reflection,chang2023approximately,east2023all}. 
For example, GQML has been used to classify images with spatial symmetries~\cite{west2023reflection,chang2023approximately}, and more general data which 
satisfies permutation symmetries~\cite{schatzki2022theoretical, sauvage2022building, heredge2023permutation}.
By constraining the search space of the optimisation procedure, such symmetry-informed models have
in a few instances been proved to be free of barren plateaus~\cite{schatzki2022theoretical, pesah2021absence}.
Despite a general theory of the trainability of QML models starting to emerge~\cite{ragone2023unified,fontana2023adjoint}, however,
the number of explicit examples of provably trainable models that have been developed and benchmarked on real data to 
test their performance in practice remains limited.
\\

\noindent
In this work, we introduce a family of rotationally equivariant QML models for classifying two dimensional data with labels that are invariant to rotations by $2\pi/2^k$ for $k\in\mathbb{N}$.
In particular, these models are applicable to image data, the consideration of the symmetries of which in the GQML context 
has previously been limited to
reflections~\cite{west2023reflection} and 90\textdegree\ rotations~\cite{chang2023approximately}. 
Our proposed architectures allow for the fraction of resources allocated to the processing of the radial and angular degrees of freedom to be controlled,
enabling for an interpolation between a model with no explicit notion of rotational symmetry $(k=0)$ to models which respect continuous rotations
$(k\to\infty)$. 
This is facilitated by choosing an encoding strategy which ensures that the representation of the group action on the encoded quantum states is that 
of the regular representation of $\mathbb{Z}_{2^k}$, which can be diagonalised via a (quantum) Fourier transform, following which it
becomes simple to write down an equivariant model, without needing to resort to explicit symmetrisation techniques such as twirling~\cite{ragone2022representation}.
We study the trainability of our models, finding that
the scaling of the variance of the cost function with respect to the trainable parameters of the model differs 
qualitatively depending on the fraction of the qubits used to encode the radial and angular degree of freedom.
Specifically, we use recently developed techniques~\cite{ragone2023unified,fontana2023adjoint} which involve studying the dynamical Lie algebras of 
quantum circuits to prove that 
our models are free from barren plateaus when the number of qubits devoted to encoding the radial information grows at most logarithmically with the total 
number of qubits. \\

\begin{figure*}
\makebox[\textwidth][c]{\includegraphics[width=1.01\textwidth]{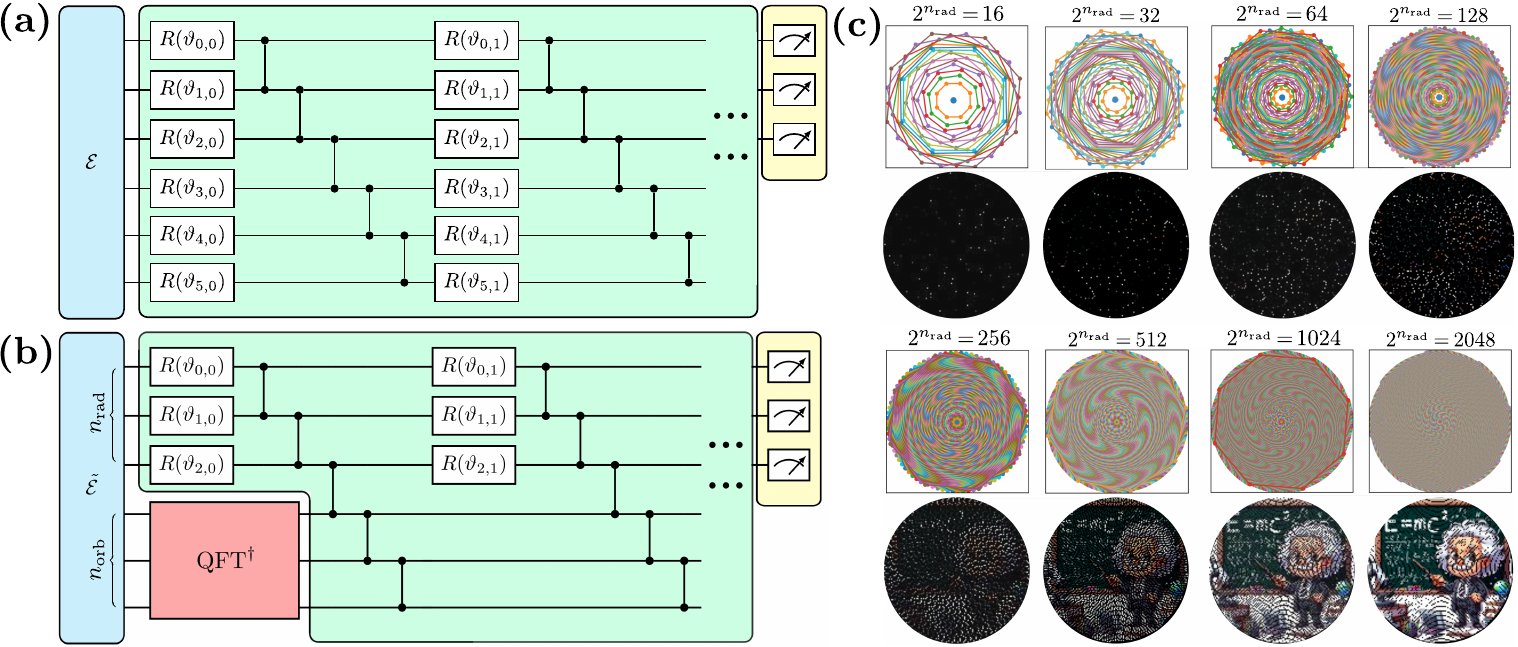}}
  \caption{\textbf{Rotationally Equivariant Quantum Machine Learning.} (a) A generic fully expressible QML model. Models such as 
  these are known to suffer from \textit{barren plateaus},
  regions in their training landscape where the gradients of the cost function vanish exponentially, rendering the models untrainable.
  (b)  The architecture of our rotation equivariant models, which are less than fully expressible,
  having been restricted to the subspace of the full Hilbert space which respects rotation symmetry.
  The combination of our encoding choice $\tilde{\mathcal{E}}$ (see Equation~\ref{eq:encoded_state}) and the quantum Fourier transform 
  ensures that the induced 
  representation of the rotation symmetry group 
  will commute with any operation on the first $n_{\mathrm{rad}}$ qubits, and operations on the last $n_{\mathrm{orb}}$
  qubits which are diagonal in the computational basis. Thus by employing an arbitrary circuit on the 
  first $n_{\mathrm{rad}}$ qubits and CZ gates only on the remainder, we construct an equivariant model.
  (c) Our choice of encoding leads to pixels being sampled from images at the vertices of $2^{n_{\mathrm{rad}}}$
  regular $2^{n_{\mathrm{orb}}}$-gons (top row), with examples of the resulting reconstructed images  for a  
  128$\cp 128$ pixel image with $2^{n_{\mathrm{orb}}}=8$ (i.e. octagons) and $16\leq 2^{n_{\mathrm{rad}}}\leq 2048$ (bottom row). 
  }
\label{fig:1}
\end{figure*}

\noindent
In addition to our analytical results, we test the performance of our models in practice by benchmarking them against generic, fully expressible models 
on a 10-class classification task involving the analysis of scanning tunnelling microscope (STM) images of phosphorus impurities implanted in silicon.
This dataset was chosen as the classification of such STM images
is a relatively complicated problem which possesses rotational symmetry and also has important applications to the 
fabrication of spin-based quantum computers in silicon~\cite{usman2016spatial}.
We find that our rotationally equivariant models achieve drastically better results and scale much more readily to deeper model sizes than 
the generic models, which suffer from significant trainability issues. 
This is achieved despite the generic models
being in principle capable of finding any solution available to the 
equivariant models, and adds to the growing body of evidence supporting GQML as a promising pathway for quantum machine learning.

\section{Rotationally-Equivariant Quantum Machine Learning}
We begin by briefly introducing the basic ideas of GQML needed 
to motivate our rotationally equivariant architectures.
Much more detail can be found in the literature, e.g. Refs.~\cite{ragone2022representation, nguyen2022theory}.
In the standard supervised learning setting we are given data $\x$ drawn from a set $\mathcal{X}$, and attempt to determine associated labels  
$y(\x)\in\mathcal{Y}=\{1,2,\ldots,n_{\mathrm{classes}}\}$.
The quantum models we employ for this task here are comprised of three main components --  an initial data encoding layer, a trainable variational section, and a final set of measurements 
(see Figure~\ref{fig:1}(a,b)). In the final stage $ n_{\mathrm{classes}}$ operators $\{M_j\}_{j=1}^{n_{\mathrm{classes}}}$ are measured, 
with the prediction $\hat{y}$ of the model on the input $\x$ defined to be the class corresponding to the operator with the highest expectation value, i.e. 
\begin{equation}
  \hat{y}_{\theta}(\x) = \argmax_j \bra{0}\mathcal{E}(\x)^\dagger U_{\theta}^\dagger M_j U_{\theta}\mathcal{E}(\x)\ket{0} 
  \label{eq:pred}
\end{equation}
where $\mathcal{E}(\x)$ denotes the data  encoding unitary for the input $\x$, and $U_\theta$ the variational component evaluated with 
parameters $\theta\in\Theta$. 
Now, if $\mathcal{G}:\mathcal{X}\to\mathcal{X}$ is a symmetry of the data at the level of the classification labels, i.e. 
$y(g(\x))=y(\x)\ \forall\x\in\mathcal{X},\ g\in\mathcal{G}$ (but \textit{not} necessarily $g(\x)=\x$),
then we wish to build models with inductive biases which automatically enforce respect for this symmetry, i.e. 
$\hat{y}_{\theta}(\x)=\hat{y}_\theta(g(\x))\ \forall\x ,\theta$.
The encoding map $\mathcal{E}:\mathcal{X}\to\mathcal{H}$ induces a unitary representation $R$ of the symmetry group $\mathcal{G}$ 
on $\mathcal{H}$ making the following diagram commute: 
\[
\begin{tikzcd}[
    row sep=huge,
    column sep=huge,
    cells={nodes={font=\normalsize, inner sep=1ex, outer sep=.5ex}}
]
  \mathcal{X} \arrow[r, "g"] \arrow[d, "\mathcal{E}"', swap]
  & \mathcal{X} \arrow[d, "\mathcal{E}"] \\
  \mathcal{H} \arrow[r, "R(g)"]
  & \mathcal{H}
\end{tikzcd}
\]
and the prediction of the model on a data point acted upon by $g\in\mathcal{G}$ becomes
\begin{equation}
 \hat{y}_{\theta}(g(\x)) = \argmax_j \bra{\psi(\x)}R_g^\dagger U_{\theta}^\dagger M_j U_{\theta}R_g\ket{\psi(\x)}  
  \label{eq:pred_sym}
\end{equation}
with $\ket{\psi(\x)}=\mathcal{E}(\x)\ket{0}$ the encoded state of the data.
For the predictions of our model to be invariant with respect to the action of the symmetry, then, equating Equations \ref{eq:pred} and \ref{eq:pred_sym}
we find a condition for equivariance,
\begin{equation}
  \left[ R(g), U_{\theta}^\dagger M_j U_{\theta} \right] = 0\qquad \forall g\in\mathcal{G},\ \theta\in\Theta,\ j\in\{1,\ldots, n_{\mathrm{classes}}\}  
  \label{eq:equi_comm}
\end{equation}
i.e. $U_{\theta}^\dagger M_j U_{\theta} \in \mathfrak{comm}(R)$, and the representation of the symmetry group enforces constraints on the 
architecture of the model. 
The subspace $\mathfrak{comm}(R)$ will depend on the nature of 
the encoding map $\mathcal{E}$, to which we now turn for our specific application.\\

Our encoding strategy  $\tilde{\mathcal{E}}$ is a slight modification of the standard amplitude encoding method.
For an image this would typically involve flattening the image into a one dimensional vector $\x$ and then 
mapping the elements of that vector to the amplitudes of a set of basis states, i.e.
$\x\mapsto 1/\norm{\x}^2\sum_i x_i\ket{i}$.
This, however, would lead to complicated and non-local elements $R(g)$ of the representation of the  symmetry group.
In principle this is not a fundamental issue, and one could construct an equivariant model via (for example)
``twirling'' (i.e. explicitly projecting onto the invariant subspaces of the representation)~\cite{ragone2022representation}, 
but here we instead opt to implement an encoding method more amenable to rotational symmetry, 
which will allow us to easily write down an equivariant model in terms of familiar quantum gates.
Expending some effort modifying the initial data encoding stage in order to simplify the equivariance constraints has previously 
been successfully employed in the context of reflection symmetry~\cite{west2023reflection}.
Our strategy is parameterised by two hyper-parameters, $n_{\mathrm{rad}}$ and $n_\mathrm{orb}$
(respectively the number of qubits assigned to encoding radial and orbital degrees of freedom, 
with $n=n_{\mathrm{rad}}+n_\mathrm{orb}$ the total number of qubits), and 
involves constructing $2^{n_\mathrm{rad}}$ regular $2^{n_{\mathrm{orb}}}$-gons
(see Figure \ref{fig:1}(c) for examples with  $n_\mathrm{orb}=3$ and $4\leq n_\mathrm{rad}\leq 11$).
The image is sampled from at each vertex of each polygon, and, 
indexing the  polygons by $r\in\{0,1,\ldots,2^{n_{\mathrm{rad}}}-1\}$, and the angular coordinate describing the position on the polygon by
$\phi=2\pi k/2^{n_\mathrm{orb}},\ k\in\{0,1,\ldots,2^{n_\mathrm{orb}}-1\}$, the corresponding pixel values $x_{r,\phi}$ are used to construct the state
\begin{equation}
  \ket{\psi(\x)} = \frac{1}{\norm{\x}^2} \sum_{r,\phi} x_{r,\phi}\ket{r,\phi}   
  \label{eq:encoded_state}
\end{equation}
With this choice of data encoding the representation $R(g)$ of the operation of rotating by $2\pi g/2^{n_\mathrm{orb}}$ simply acts as
\begin{equation}
 R(g)\ket{\psi(\x)} = \sum_{r,\phi} x_{r,\phi}\ket{r,\phi + 2\pi g/2^{n_\mathrm{orb}}\ (\mathrm{mod}\ 2\pi)} 
\end{equation}
\noindent
Explicitly, with respect to this basis we have (in the case $n_\mathrm{orb}=2$, for example)
\[
  R(0) = \mathbb{I}_{\mathrm{rad}}\otimes \begin{pmatrix}
    1 & 0 & 0 &0  \\
    0 & 1 & 0 &0  \\
    0 & 0 & 1 &0  \\
    0 & 0 & 0 &1  \\
\end{pmatrix},\qquad
R(1) = \mathbb{I}_{\mathrm{rad}}\otimes\begin{pmatrix}
    0 & 0 & 0 &1  \\
    1 & 0 & 0 &0  \\
    0 & 1 & 0 &0  \\
    0 & 0 & 1 &0  \\
\end{pmatrix}
\]
\[ 
R(2) = \mathbb{I}_{\mathrm{rad}}\otimes\begin{pmatrix}
    0 & 0 & 1 &0  \\
    0 & 0 & 0 &1  \\
    1 & 0 & 0 &0  \\
    0 & 1 & 0 &0  \\
\end{pmatrix},\qquad
R(3) =\mathbb{I}_{\mathrm{rad}}\otimes \begin{pmatrix}
    0 & 1 & 0 &0  \\
    0 & 0 & 1 &0  \\
    0 & 0 & 0 &1  \\
    1 & 0 & 0 &0  \\
\end{pmatrix}
\]

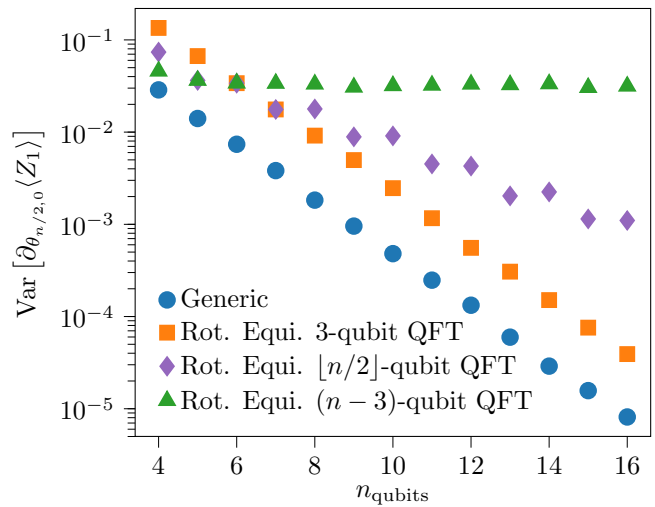
\begin{figure}
\begin{tikzpicture}

\definecolor{darkgray176}{RGB}{176,176,176}
\definecolor{darkorange25512714}{RGB}{255,127,14}
\definecolor{forestgreen4416044}{RGB}{44,160,44}
\definecolor{mediumpurple148103189}{RGB}{148,103,189}
\definecolor{steelblue31119180}{RGB}{31,119,180}

\begin{axis}[
legend cell align={left},
legend style={
  fill opacity=1.0,
  draw opacity=1,
  text opacity=1,
  at={(0.03,0.03)},
  anchor=south west,
  draw=none,
  fill=none
},
log basis y={10},
tick align=outside,
tick pos=left,
x grid style={darkgray176},
xlabel={\(\displaystyle n_{\mathrm{qubits}}\)},
xmin=3.4, xmax=16.6,
xtick style={color=black},
y grid style={darkgray176},
ylabel={\(\displaystyle  \mathrm{Var}\left[\partial_{\theta_{n/2,0}}\langle Z_1\rangle\right]\)},
ylabel style={yshift=5pt},
ymin=5.00995873381322e-06, ymax=0.219156988196939,
ymode=log,
ytick style={color=black}
]
\addplot [semithick, steelblue31119180, mark=*, mark size=3, mark options={solid}, only marks]
table {%
4 0.0287082695452347
5 0.0140485473747844
6 0.00738444241462472
7 0.00382441086566244
8 0.00182996942707149
9 0.000956524876971033
10 0.000480711860369096
11 0.000248013066990622
12 0.000132632642147355
13 5.97532924559757e-05
14 2.89849698213173e-05
15 1.57220064782744e-05
16 8.14301268187451e-06
};
\addlegendentry{Generic}
\addplot [semithick, darkorange25512714, mark=square*, mark size=2.7, mark options={solid}, only marks]
table {%
4 0.134835534462253
5 0.0667644917711864
6 0.0339421156237326
7 0.0176117362182932
8 0.00916319046901668
9 0.00497329604895582
10 0.00246784115436452
11 0.00116516496243061
12 0.000555693507932158
13 0.000306753993900294
14 0.000150785695892405
15 7.58045207673048e-05
16 3.9156385694753e-05
};
\addlegendentry{Rot. Equi. $3$-qubit QFT}
\addplot [semithick, mediumpurple148103189, mark=diamond*, mark size=3.5, mark options={solid}, only marks]
table {%
4 0.073732301400117
5 0.0363476803147932
6 0.0339421156237326
7 0.0176117362182932
8 0.0178290180546564
9 0.00888316210623914
10 0.00909623627881679
11 0.00451619072409216
12 0.00428470583496594
13 0.00203287413973089
14 0.00224442501050539
15 0.00114389809487631
16 0.00110104637249665
};
\addlegendentry{Rot. Equi. $\lfloor{n/2}\rfloor$-qubit QFT}
\addplot [semithick, forestgreen4416044, mark=triangle*, mark size=3.5, mark options={solid}, only marks]
table {%
4 0.0457770923533566
5 0.0363476803147937
6 0.0339421156237331
7 0.0337519157280206
8 0.0331137678798169
9 0.0306734700940407
10 0.0318269256032519
11 0.0321885719046995
12 0.0330653753113584
13 0.0326104456536457
14 0.0333470173336821
15 0.0304038366748015
16 0.0314651369788205
};
\addlegendentry{Rot. Equi. $(n-3)$-qubit QFT}
\end{axis}

\end{tikzpicture}
  \caption{\textbf{Gradient Scaling}. The variance of the derivative of a measured observable with respect to a 
  parameter in the middle of the model under different assignments of the qubits to encoding radial and rotational information. 
  In all cases the variances are estimated by averaging 
  over 1000 randomly initialised models.  When the number $n_{\mathrm{rad}}$ of qubits not involved in the quantum Fourier transform 
  grows faster than 
  logarithmically with $n_{\mathrm{qubits}}$ the model suffers from a barren plateau, as exhibited by the exponential 
  vanishing of the cost function gradients.
  }
\label{fig:bp}
\end{figure}

\begin{figure*}
\includegraphics[width=\textwidth]{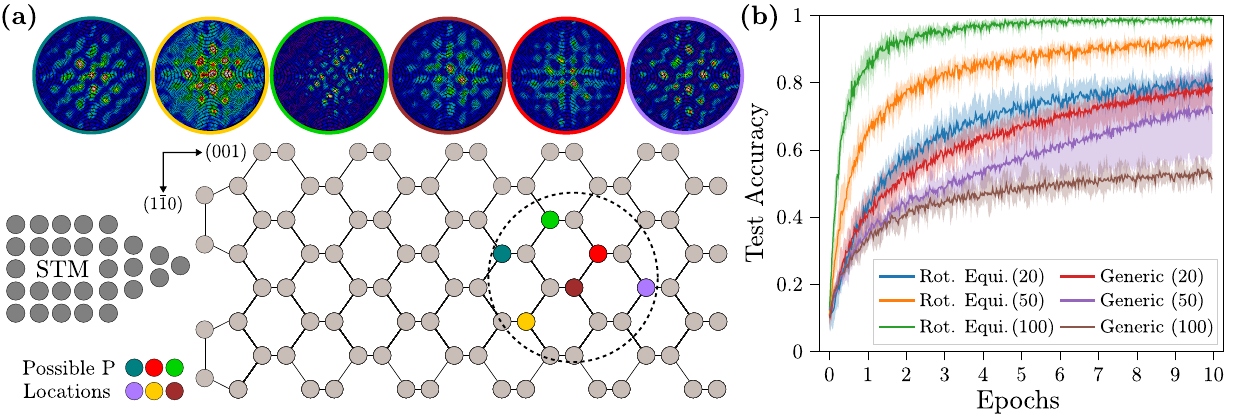}
  \caption{\textbf{Pinpointing Atomic Impurities with Symmetry-Informed QML.} 
  In order to numerically test the performance of our models in practice we consider a
  classification task which involves determining the 
  location of phosphorus donor atoms implanted into silicon by analysing their STM images. 
  (a) P donors in Si create remarkably rich STM images which depend strongly on their precise location in the lattice.
  Here we consider the problem of determining the correct lattice site from ten possibilities 
  corresponding to positions between 3$a_0$ and 4.5$a_0$ beneath the surface (with $a_0\approx 0.54$nm the silicon lattice constant)
  given a (simulated, see Appendix B) STM image.
  (b) 
  We randomly initialise and train each model configuration five times, plotting the mean accuracy achieved on a validation set of 500 examples
  throughout training, and shading the area between the best and worst performing of the training runs.
  The architecture and number of layers in each model is indicated.
  Our rotationally equivariant models drastically outperform their generic counterparts on this task, despite possessing fewer trainable 
  parameters and having access to only a fraction of the total Hilbert space. 
  Moreover, we find that increasing the depth of the generic models actually adversely affects their performance in practice, 
  due to the increasing issues with trainability and given only a finite set of training examples.
  The training set consists of 7500 examples, equally distributed amongst the classes.
  }
\label{fig:2}
\end{figure*}

\noindent
where $\mathbb{I}_r$ is the identity matrix acting on the radial degrees of freedom.
Thus the rotations act on the angular degrees of freedom as the regular representation of the additive group $\mathbb{Z}_{2^{n_\mathrm{orb}}}$,
and can be diagonalised by the Fourier transformation over that group~\cite{childs2010quantum}. With $\mathbb{Z}_{2^{n_\mathrm{orb}}}$ 
being an abelian group, this is simply the familiar discrete Fourier transform, and can be implemented
efficiently on a quantum computer by means of the quantum Fourier transform (QFT).
Operators $ O\in\mathfrak{comm}(R)$ are then also block diagonalised by the QFT 
(in fact fully diagonalised in this case as the irreducible representations of abelian groups are one dimensional), and so 
following the initial operations 
$\left(\mathbb{I}_{\mathrm{rad}}\otimes \mathrm{QFT}^\dagger _{\mathrm{orb}}\right)\circ \tilde{\mathcal{E}}(\x)$
we construct our equivariant model from components which are 
diagonal on the angular degrees of freedom and unconstrained
on the radial degrees of freedom. 
In this work for simplicity we take the diagonal component to be chains of nearest-neighbour CZ gates (see Figure~\ref{fig:1}(b))
and employ a fully expressible ansatz on the radial component,
although we could also include (for example) parameterised $Z$ rotations on the angular register, or introduce additional constraints on 
the radial register.\\

We now turn to the trainability of our rotationally equivariant models. 
Our main analytical result is that, under some conditions, the models are free from barren plateaus.
These results utilise the recent insights of Refs.~\cite{ragone2023unified,fontana2023adjoint},
which analyse the mean and variance of the cost function of QML models by studying the Lie closure of the generators of the circuit.
In brief, given a layered quantum circuit of the form
\begin{equation}
  \mathcal{U}(\boldsymbol{\theta}) =\prod_{i=1}^L U(\boldsymbol{\theta}_i)=\prod_{i=1}^L\prod_{j=1}^J e^{-i\theta_{ij}H_{j}}
  \label{eq:layers}
\end{equation}
one can relate the cost function mean and variance to the dimension of the \textit{dynamical Lie algebra} 
$ \mathfrak{g} = \mathrm{span}\ \langle iH_0,iH_1,\ldots, iH_J\rangle_{\mathrm{Lie}}$, 
where $ \langle \mathcal{G}\rangle_{\mathrm{Lie}}$ is the set formed from repeated nested commutators of the elements of $\mathcal{G}$.
In particular, if $\mathrm{dim\ }  \mathfrak{g}$ grows exponentially with the number of qubits then the model necessarily exhibits barren plateaus~\cite{ragone2023unified,fontana2023adjoint}.
Intuitively, the relevance of the nested commutators of the generators can be seen by applying the Baker–Campbell–Hausdorff formula to Equation~\ref{eq:layers}.
In this work we show that the dimension of the dynamical Lie algebra of our models is exponential in $n_{\mathrm{rad}}$, but only linear in 
$n_{\mathrm{orb}}$. Specifically:

\begin{prop} \label{prop:dim}
  The dynamical Lie algebra $\mathfrak{g}$ of our rotationally equivariant model has dimension
\[ \mathrm{dim}\ \mathfrak{g} = 2\cdot 4^{n_{\mathrm{rad}}} + n_{\mathrm{orb}} -1 \]
  where $n_{\mathrm{rad}}$ ($n_{\mathrm{orb}}$) is the number of qubits used to encode the radial (angular) degrees of freedom.
\end{prop}
\addtocounter{prop}{-1}

\noindent
The exponential dependence on $n_{\mathrm{rad}}$ is a result of us employing a fully expressible ansatz on the radial register;
one could replace it with a variational structure possessing a DLA growing only polynomially in $n_{\mathrm{rad}}$, 
but this constraint is not forced by rotational equivariance.
Combined with the results of Ref.~\cite{ragone2023unified} 
our main result follows:

\begin{corol} \label{corol:bp}
  Our rotationally equivariant models are free from barren plateaus if and only if 
  $4^{-n_{\mathrm{rad}}}\mathcal{P}_{\mathfrak{g}_s}(\rho)\notin\mathcal{O}(1/b^n)$ for any $b>2$, where $\rho=\ket{\psi(\x)}\bra{\psi(\x)}$ is the initial state and 
  $\mathcal{P}_{\mathfrak{g}_s}(\rho)$ its purity  with respect to the semisimple component of the dynamical Lie algebra $\mathfrak{g}$. 
\end{corol}
\addtocounter{corol}{-1}

\noindent
Here $\rho_\mathfrak{g}$ denotes the orthogonal projection of $\rho$ onto $\mathfrak{g}$ (with respect to the Hilbert-Schmidt inner product)
with purity $\mathcal{P}_{\mathfrak{g}}(\rho) := \tr [(\rho_{\mathfrak{g}})^2]$.
Proofs of Proposition~\ref{prop:dim} and Corollary~\ref{corol:bp} are given in Appendix A.
Corollary~\ref{corol:bp} gives the constraint within which the rotationally equivariant models must operate
in order to not suffer from barren plateaus, and has several immediate consequences.
Firstly, regardless of the choice of $n_{\mathrm{rad}}$, the purity $\mathcal{P}_{\mathfrak{g}_s}(\rho)$ cannot vanish faster than $2^{-n}$.
Secondly, as $\mathcal{P}_{\mathfrak{g}}(\rho)\leq 1$, the number of qubits used in the radial 
register is bounded by 
$n_{\mathrm{rad}}\leq n/2 + o(n)$. 
These restrictions correspond respectively to requirements on (a generalised notion of)
the inital state entanglment~\cite{cerezo2021cost} and the circuit expressibility~\cite{holmes2022connecting}, 
and are linked through the unifying Lie-algebraic theory of Ref.~\cite{ragone2023unified}.
In Figure~\ref{fig:bp} we plot empirical results for the cost function gradient scaling under various choices of $n_{\mathrm{rad}}/n_{\mathrm{qubits}}$, 
the conclusions of which are consistent with our analytical findings.
Under the constraints of Corollary~\ref{corol:bp}, then,
our rotationally equivariant models join the (short) list of QML architectures provably free from barren plateaus~\cite{schatzki2022theoretical, pesah2021absence}.
This is an extremely encouraging finding, but nonetheless the true test of machine learning models is to benchmark their performance 
in practice on a real dataset, which is the task to which we now turn.

\section{Numerical Results}
We test our models on a 10-class dataset consisting of STM images of phosphorus (P) donor impurities implanted in silicon (Si),
with the classification task being to identity, with lattice-site precision, the location of the donor which produced the STM image (see Figure~\ref{fig:2}).
This is enabled in principle by the remarkable sensitivity of the images to the precise location of the donor in the Si lattice~\cite{usman2016spatial,west2021influence},
and has important applications to the construction of high-fidelity spin-based quantum computers in silicon~\cite{usman2016spatial}.
This problem has been previously been tackled with machine learning approaches~\cite{usman2020framework,west2022framework}, under the assumption of a fixed orientation of the Si crystal with 
respect to the direction of motion of the STM. Relaxing that assumption leads naturally to STM images which are randomly rotated, forming
the classification task we consider here.\\

We simulate the STM images resulting from P donors implanted at ten different positions between $3a_0$ and $4.5a_0$ below the surface (with $a_0\approx$ 0.54nm the Si lattice constant)
using a multimillion atom tight-binding simulation implemented within \texttt{NEMO-3D}~\cite{nemo} (see Appendix B for details). 
This results in $128\cp 128$ pixel grayscale images, to which we add Gaussian noise and rotate, before 
encoding them into quantum states by means of Equation~\ref{eq:encoded_state} with $n_{\mathrm{rad}}=10$
and $n_{\mathrm{orb}}=3$. This results in us sampling from the images at the vertices of 1024 octagons 
(see Figure~\ref{fig:1}(c) for the general strategy and Figure~\ref{fig:2}(a) for reconstructions of the resulting STM images (without noise)).
We also consider a generic model which employs standard amplitude encoding ($\x\mapsto 1/\norm{\x}^2\sum_i x_i\ket{i}$) and a maximally expressible ansatz
(see Figure~\ref{fig:1}(b)).
The models are implemented within Pennylane~\cite{bergholm2018pennylane} and trained with the \texttt{ADAM} optimiser~\cite{adam} with a learning rate of 1e-3 
on a dataset of 7500 images over 10 epochs, with their accuracy on a validation set of 500 examples recorded throughout the training process.
The loss funtion $\ell_{\theta}$ for an input state $\rho=\ket{\psi(\x)}\bra{\psi(\x)}$ with true label $y\in\{1,2,\ldots,n_{\mathrm{classes}}\}$ 
is given by 
\begin{equation}
  \ell_{\boldsymbol{\theta}}(\rho, Z_y) = - \tr\left[ U_{\boldsymbol{\theta}}\rho U_{\boldsymbol{\theta}}^\dagger Z_y \right]
  \label{eq:loss}
\end{equation}
where $Z_y$ denotes the Pauli $Z$ gate acting on the $y$th qubit.
The results, for models with 20, 50, and 100 layers are plotted in Figure~\ref{fig:2}(b). 
We train each model five times from different random initialisations, and plot the average over the training runs.
The shaded areas are bounded by the best and worst of the runs.
We find that the rotationally equivariant models drastically outperform their generic counterparts, which suffer from trainability issues that only worsen 
within increasing model depth.
Combined with our analytical results, we conclude that these models offer a significant advantage over a generic ansatz for two 
dimensional data with rotational symmetry.

\section{Summary and Outlook}
Much as early classical neural networks suffered from significant trainability issues~\cite{bengio1994learning} 
that were solved only after 
considerable experimentation and a sequence of increasingly optimised architectural decisions, early QML models 
based on generic circuits look likely to be replaced by carefully designed constructions.
In the absence of large-scale fault tolerant quantum computers on which to easily experiment with
different types of QML architectures in the many qubit regime, symmetry principles represent a promising guide to current research~\cite{meyer2022exploiting,ragone2022representation,nguyen2022theory,schatzki2022theoretical,sauvage2022building,skolik2022equivariant,larocca2022group,zheng2022benchmarking,west2023reflection,chang2023approximately,east2023all}
that has now lead to the discovery of multiple architectures which are provably trainable~\cite{schatzki2022theoretical, pesah2021absence}, 
including in this work.
Our results have been facilitated by the increased understanding of the nature of QML models afforded by their 
rapidly developing Lie-algebraic theory~\cite{ragone2023unified,fontana2023adjoint},
through which we can study the complexity of our models as they interpolate between
the trainable and untrainable regimes, as characterised by the dimensionality of the dynamical Lie algebra.
Finally we note that imposing restrictions on the expressibility of QML models introduces the possibility of efficient 
classical simulations, work on which in the Lie algebra context has already begun~\cite{goh2023lie}.
Families of models such as ours which can interpolate between easily trainable but potentially classically simulable circuits 
through to classically intractable but barren plateau ridden circuits then constitute an interesting research direction 
in the ongoing search for practical quantum advantage in machine learning.\\ 

\noindent
\textbf{Acknowledgements:} M.W. acknowledges the support of Australian Government Research Training Program Scholarships. 
This work was supported by Australian Research Council Discovery Project DP210102831. 
Computational resources were provided by the National Computing Infrastructure (NCI) and the Pawsey Supercomputing Research Center 
through the National Computational Merit Allocation Scheme (NCMAS).
This research was supported by The University of Melbourne’s Research Computing Services and the Petascale Campus Initiative.\\

\noindent
\textbf{Code availability:}
The code which supports the findings of this article is available at \url{https://github.com/maxwest97/rotation-equivariant-qnn}

\vspace{20mm}
\section*{Appendix A}
In this Appendix we prove Proposition~\ref{prop:dim} and Corollary~\ref{corol:bp} from the main text.
\begin{prop} 
  The dynamical Lie algebra $\mathfrak{g}$ of our rotationally equivariant model has dimension
\[ \mathrm{dim}\ \mathfrak{g} = 2\cdot 4^{n_{\mathrm{rad}}} + n_{\mathrm{orb}} -1 \]
  where $n_{\mathrm{rad}}$ ($n_{\mathrm{orb}}$) is the number of qubits used to encode the radial (angular) degrees of freedom.
\end{prop}

\begin{proof}[Proof of Proposition \ref{prop:dim}]
The dynamical Lie algebra $\mathfrak{g}$ (DLA)~\cite{larocca2022diagnosing,ragone2023unified,fontana2023adjoint}  of a layered quantum circuit of the form 
\begin{align}
  \mathcal{U}(\boldsymbol{\theta}) &= \prod_{i=1}^L U(\boldsymbol{\theta}_i)\\
  &=\prod_{i=1}^L\prod_{j=1}^J e^{-i\theta_{ij}H_{j}}
\end{align}
is given by the Lie closure of the generators $\mathcal{G}$ of its layers,
\begin{equation}
 \mathfrak{g} =\mathrm{span}\ \langle i\mathcal{G} \rangle_{\mathrm{Lie}} = \mathrm{span}\ \langle iH_0,iH_1,\ldots, iH_J\rangle_{\mathrm{Lie}} 
\end{equation}
where $ \langle \mathcal{G}\rangle_{\mathrm{Lie}}$ is the set formed from repeated nested commutators of the elements of $\mathcal{G}$.
In our rotationally-equivariant case the generators $H_j$ of the circuit are the operators that when exponentiated produce arbitrary rotations on the first 
$n_{\mathrm{rad}}$ qubits, and $CZ_{i,i+1}$ gates between nearest-neighbour qubits. 
We thus have that $X_i,Y_i,Z_i\in\mathcal{G}\ \forall i\leq n_{\mathrm{rad}}$, as well as $cz_{i,i+1}$ where 
$\exp -i(cz_{i,i+1})=CZ_{i,i+1}\ \forall i \leq n_{\mathrm{qubits}} - 1$. Explicitly, in the computational basis in the two qubit case we have
\[ CZ_{1,2}= \begin{pmatrix}
1 & 0 & 0 & 0 \\
0 & 1 & 0 & 0 \\
0 & 0 & 1 & 0 \\
0 & 0 & 0 & -1 \\
\end{pmatrix} = \exp \begin{pmatrix}
0 & 0 & 0 & 0 \\
0 & 0 & 0 & 0 \\
0 & 0 & 0 & 0 \\
0 & 0 & 0 & i\pi \\
\end{pmatrix} =\exp -i(cz_{1,2}) 
\]
from which we can readily obtain expressions for the $CZ$ generators in terms of Pauli operators,
\[cz_{i,i+1} =- \frac{\pi}{4}\left( \mathbb{I} - Z_i - Z_{i+1} + Z_iZ_{i+1} \right) \]
We therefore wish to determine the dimension of the vector space spanned by the Lie closure of the set
\[ \left\{ X_i,\ Y_i,\ Z_i \right\}_{i=1}^{i=n_{\mathrm{rad}}}\ \bigcup \ \left\{ \mathbb{I} - Z_i - Z_{i+1} + Z_iZ_{i+1} \right\}_{i=1}^{i=n_{\mathrm{qubits}}-1} \]
First we restrict our attention to $i\leq n_{\mathrm{rad}}$. We have
\[ [X_i, \mathbb{I} - Z_i - Z_{i+1} + Z_iZ_{i+1} ]= 2i\left(Y_i - Y_iZ_{i+1}\right)  \]
\[ [Y_i, \mathbb{I} - Z_i - Z_{i+1} + Z_iZ_{i+1} ]= -2i\left(X_i - X_iZ_{i+1}\right)  \]
So $\{X_iZ_{i+1},\ Y_iZ_{i+1}\}_{i\leq n_{\mathrm{rad}}}\subset \mathfrak{g}$. By commuting these terms with appropriate single-qubit operators 
we can generate all two-qubit nearest-neighbour operators acting within the first $n_{\mathrm{rad}}$ qubits. 
It then follows by the argument presented in Appendix I of Ref.~\cite{larocca2022diagnosing} 
that all $4^{n_{\mathrm{rad}}}$ of the $n$-body Pauli strings with support on the first $n_{\mathrm{rad}}$ qubits lie within the DLA
  (giving complete expressibility on those qubits as $\mathrm{dim}\ \mathfrak{u}(n_{\mathrm{rad}})=4^{n_{\mathrm{rad}}}$). 
Furthermore, for each of these operators we can construct (with the aid of $cz_{n_{\mathrm{rad}},n_{\mathrm{rad}}+1}$)
a corresponding operator which also applies $Z_{n_{\mathrm{rad}}+1}$, doubling the number of linearly independent operators
we have found within $\mathfrak{g}$.
Finally, the remaining $ n_{\mathrm{orb}} -1$ generators of $CZ$ gates commute both amongst themselves and with all 
the other linearly independent operators in the DLA, bringing the total dimension to
  $2(4^{n_{\mathrm{rad}}}) + n_{\mathrm{orb}} -1$. 
\end{proof}

\noindent
Our next result makes use of the notion of the $\mathfrak{g}$-purity $\mathcal{P}_{\mathfrak{g}}(O) = \tr O_\mathfrak{g}^2$ of an operator $O$ with respect to an operator 
subalgebra $\mathfrak{g}$~\cite{ragone2023unified,somma2004nature}, 
with $O_\mathfrak{g}$ the orthogonal (with respect to the Hilbert-Schmidt norm) projection of $O$ onto $\mathfrak{g}$.

\vspace{10mm}
\begin{corol}
  Our rotationally equivariant models are free from barren plateaus if and only if 
  $4^{-n_{\mathrm{rad}}}\mathcal{P}_{\mathfrak{g}_s}(\rho)\notin\mathcal{O}(1/b^n)$ for any $b>2$, where $\rho=\ket{\psi(\x)}\bra{\psi(\x)}$ is the initial state and 
  $\mathcal{P}_{\mathfrak{g}_s}(\rho)$ its purity  with respect to the semisimple component of the dynamical Lie algebra $\mathfrak{g}$. 
\end{corol}

\begin{proof}[Proof of Corollary \ref{corol:bp}]
  Being a subalgebra of $\mathfrak{u}\left(2^{n_{\mathrm{qubits}}}\right)$, $\mathfrak{g}$ is a reductive Lie algebra, meaning that  
  we have a (unique up to rearrangements of the terms) decomposition~\cite{knapp1996lie}
  \begin{equation}
   \mathfrak{g}\cong\mathfrak{g}_1\oplus\mathfrak{g}_2\oplus\ldots\oplus\mathfrak{g}_k 
    \label{eq:reductive}
  \end{equation}
   where the $\mathfrak{g}_i$ are simple Lie algebras for $i<k$, and $\mathfrak{g}_k$ is the abelian centre of $\mathfrak{g}$.
  Theorem 1 of Ref.~\cite{ragone2023unified} states that, given a DLA $\mathfrak{g}$ in the form of Equation~\ref{eq:reductive},
  an input state $\rho$ and a loss function of the form $\ell_{\boldsymbol{\theta}}(\rho, O) = \tr\left[ U_{\boldsymbol{\theta}}\rho U_{\boldsymbol{\theta}}^\dagger O\right]$,
  then if $O\in i\mathfrak{g}$ or $\rho\in i\mathfrak{g}$ we have:
  \begin{equation}
    \mathbb{E}_{\boldsymbol{\theta}}[\ell_{\boldsymbol{\theta}}(\rho, O)] = \tr\left[\rho_{\mathfrak{g}_k}O_{\mathfrak{g}_k}\right]
    \label{eq:mean}
  \end{equation}
 \noindent
  and
  \begin{equation}
    \mathrm{Var}_{\boldsymbol{\theta}}[\ell_{\boldsymbol{\theta}}(\rho, O)] = \sum_{j=1}^{k-1}\frac{\mathcal{P}_{\mathfrak{g}_j}(\rho)\mathcal{P}_{\mathfrak{g}_j}(O)}{\mathrm{dim}\ \mathfrak{g}_j}
    \label{eq:var}
  \end{equation}
  Recalling the loss function of our 
  rotation equivariant models (Equation~\ref{eq:loss})  we see that it is in the required form to apply
  this theorem, with $O=Z_y\in i\mathfrak{g}$.  
  In order to do this we must first determine the
  decomposition of $\mathfrak{g}$ into its semisimple and abelian components, to which we now turn.\\

  \noindent
  In the proof of Proposition~\ref{prop:dim} we saw that
  in our case $\mathfrak{g}$ is spanned by the $  2\cdot 4^{n_{\mathrm{rad}}}$ Pauli strings which apply 
  arbitrary Pauli operators on the first $n_{\mathrm{rad}}$ qubits and either $I$ or $Z$ on the $(n_{\mathrm{rad}}+1)$th qubit,
  along with the generators of the final $n_{\mathrm{orb}} -1$ CZ gates. 
  Pulling out the operators $I^{\otimes (n_{\mathrm{rad}}+1)}$ and $I^{\otimes n_{\mathrm{rad}}}\otimes Z$ 
  (which commute with all elements of $\mathfrak{g}$), we find that $\mathfrak{g}$ consists of an
  $n_{\mathrm{orb}}+1$ dimensional centre, 
  along with the operators in the set 
  $
   \mathrm{span}\{A\otimes I\otimes I^{\otimes (n_{\mathrm{orb}}-1)}|\ A\in \mathfrak{su}(2^{n_{\mathrm{rad}}})\} \oplus
  \mathrm{span} \{B\otimes Z\otimes I^{\otimes (n_{\mathrm{orb}}-1)}|\ B\in \mathfrak{su}(2^{n_{\mathrm{rad}}})\} 
$.
  This does not quite give us a decomposition of the form of Equation~\ref{eq:reductive}, however, as the second of those
  summands does not form a Lie subalgebra of $\mathfrak{g}$. Equivalently, however, we can consider the set
\begin{align*}
  \mathfrak{g}_{+}\oplus\mathfrak{g}_{-}:=\ &\mathrm{span}\{A\otimes \frac{(I+Z)}{\sqrt{2}}\otimes I^{\otimes (n_{\mathrm{orb}}-1)}|\ A\in \mathfrak{su}(2^{n_{\mathrm{rad}}})\}\ \\
  \oplus\ &\mathrm{span} \{B\otimes \frac{(I-Z)}{\sqrt{2}}\otimes I^{\otimes (n_{\mathrm{orb}}-1)}|\ B\in \mathfrak{su}(2^{n_{\mathrm{rad}}})\} 
\end{align*}  
 which contains the same operators and is written in the form of a direct sum of two ideals of $\mathfrak{g}$ 
  (both of which are isomorphic to $\mathfrak{su}(2^{n_{\mathrm{rad}}})$). We therefore have the decomposition
\begin{align}
  \mathfrak{g}&\cong \mathfrak{g}_+\oplus\mathfrak{g}_-\oplus\mathfrak{g}_k \\
  &\cong \mathfrak{su}(2^{n_{\mathrm{rad}}})\oplus\mathfrak{su}(2^{n_{\mathrm{rad}}})\oplus\mathbb{R}^{n_{\mathrm{orb}}+1} 
\end{align}  
  \noindent
  and are ready to apply Equations~\ref{eq:mean} and~\ref{eq:var} for a given input state $\rho=\ket{\psi(\x)}\bra{\psi(\x)}$ with true label $y\in\{1,2,\ldots,n_{\mathrm{classes}}\}$.
  Firstly, we have that $Z_{y}\in i\mathfrak{su}(2^{n_{\mathrm{rad}}})$ has zero projection onto the centre of $\mathfrak{g}$, so 
$   \mathbb{E}_{\boldsymbol{\theta}}[\ell_{\boldsymbol{\theta}}(\rho, O)]=0$.
  To calculate the variance of the loss function we introduce orthonormal (with respect to the Hilbert-Schmidt inner product)
  bases for $\mathfrak{g}_\pm$, 
\begin{equation}
  \left\{ 2^{-n_{\mathrm{qubits}}/2} P\otimes \frac{(I\pm Z)}{\sqrt{2}}\otimes I^{\otimes (n_{\mathrm{orb}}-1)} \right\}_{P\in P_{n_{\mathrm{rad}}}\setminus \{I^{\otimes n_{\mathrm{rad}}}\}}
\end{equation}
where $P_{n_{\mathrm{rad}}}$ is the set of $n_{\mathrm{rad}}$-qubit Pauli strings.
Denoting these basis operators by $\{B_j^{(\pm)}\}_{j=1}^{4^{n_\mathrm{rad}}-1}$
we have that
the purities of $Z_y$ with respect to $\mathfrak{g}_{\pm}$ are given by
\begin{align}
  \mathcal{P}_{\mathfrak{g}_\pm}(Z_y) &= \tr \left[\sum_{j=1}^{4^{n_\mathrm{rad}}-1} \tr\left(B_j^{(\pm)\dagger} Z_y\right) B_{j}^{(\pm)}\right]^2 \\
  &=  \sum_{j=1}^{4^{n_\mathrm{rad}}-1}\abs{ \tr\left(B_j^{(\pm)\dagger} Z_y\right)}^2 \\
  &= 2^{n_{\mathrm{qubits}}-1}
\end{align}
where we used that $\tr\left(B_j^{(\pm)\dagger} Z_y\right)=0$ for 
$B_j^{(\pm)}\neq I^{\otimes (y-1)}\otimes Z \otimes I^{\otimes (n_{\mathrm{rad}}-y)}\otimes (I\pm Z)/2\otimes I^{\otimes (n_{\mathrm{orb}}-1)}$. 
From Equation~\ref{eq:var}, then, we have 
\begin{align}
  \mathrm{Var}_{\boldsymbol{\theta}}[\ell_{\boldsymbol{\theta}}(\rho, Z_y)] &= \sum_{j\in \{+,-\}}\frac{\mathcal{P}_{\mathfrak{g}_j}(\rho)\mathcal{P}_{\mathfrak{g}_j}(Z_y)}{\mathrm{dim}\ \mathfrak{g}_j}\\
  &= \frac{2^{n_{\mathrm{qubits}}-1}}{4^{n_{\mathrm{rad}}-1}}\sum_{j\in \{+,-\}}\mathcal{P}_{\mathfrak{g}_j}(\rho)\\
  &= \frac{2^{n_{\mathrm{qubits}}-1}}{4^{n_{\mathrm{rad}}-1}}\mathcal{P}_{\mathfrak{g}_s}(\rho)
\end{align}
where $\mathfrak{g}_s$ is the semisimple component of the dynamical Lie algebra $\mathfrak{g}$.
This expression vanishes exponentially in $n_{\mathrm{qubits}}$ if and only if 
$4^{-n_{\mathrm{rad}}}\mathcal{P}_{\mathfrak{g}_s}(\rho)\in\mathcal{O}\left(1/b^{n_{\mathrm{qubits}}}\right)$ for some $b>2$.
\end{proof}
\section*{Appendix B} \label{appb}
In this Appendix we discuss the relevant details of 
the simulations which produced the STM images considered in this work.
The physical principle upon which STMs rely is quantum tunnelling~\cite{chen2021introduction}. 
A metallic tip is swept across the surface to be 
imaged (at a height of order $\sim$1nm), and for a given bias voltage $V$ between the tip and surface, 
the current $I$ caused by electrons tunnelling from the surface to the tip is measured, 
producing a two dimensional image of the surface (constant height mode). 
Alternatively, one can adjust the height of the tip during the sweep so as to maintain a constant tunnelling current 
and an image can be inferred from the recorded heights (constant current mode).
The tunnelling current is given by Bardeen's tunnelling formula~\cite{bardeen}, 
\begin{equation}
  I(V) = \frac{2\pi e}{\hbar} \sum_{\mu\nu} f(E_{\mu}) (1-f(E_\nu+eV))\abs{M_{\mu\nu}}^2\delta\left(E_\mu-E_\nu-eV\right) 
\end{equation}
where $E_\mu\ (E_\nu)$ are the eigenenergies of the sample (tip), $f$ the Fermi-Dirac distribution
function and $M_{\mu\nu}$ the tunnelling matrix elements,
\begin{equation}
 M_{\mu\nu}=-\frac{  \hbar ^2  }{  2m} \int_\Sigma \big[ \chi_{\nu}^* ( \boldsymbol{r})\grad \psi_{\mu} ( \boldsymbol{r}) - \psi_{\mu} ( \boldsymbol{r})\grad \chi_{\nu}^* ( \boldsymbol{r})  \big] \vdot d\boldsymbol{S}     
\end{equation}
where $\Sigma$ is a surface separating the tip from the surface being imaged. Expanding the tip wavefunction in terms 
of the spherical harmonics $Y_{lm}$ and the spherical modified Bessel functions of the second kind $k_l$ gives,
in the notation of Ref.~\cite{mandi2015chen},
\begin{equation}
  \chi _{\nu }(\boldsymbol{r}) =\sum_{\beta} C_{\nu \beta}{k}_{\beta}(\kappa_\nu(\boldsymbol{r}-\boldsymbol{r_0}))Y_{\beta}(\theta,\phi)
\end{equation}
with $\kappa_\nu$ the vacuum decay constant and $  \beta \in \{s,p_x,p_y,p_z,d_{xy}, d_{yz}, d_{z^2 - x^2 - y^2}, d_{xz}, d_{x^2-y^2} \}$.
This decomposition yields an expression for $M_{\mu\nu}$ in terms of differential operators acting on the surface wavefunction
due to Chen~\cite{chen},
\begin{equation}
  \abs{M_{\mu\nu} }^2\propto \abs{ \sum_{\beta} C_{\nu \beta} \hat{D}_{\beta}\psi_\mu(\boldsymbol{r_0}) }^2
\end{equation}
where $\hat{D}_{\beta}$ are differential operators. 
It is known that, for the Si:P system considered here, the dominant contribution comes from the  $d_{z^2 - x^2 - y^2}$ orbital~\cite{usman2016spatial,west2021influence},
with corresponding differential operator~\cite{chen}
\[ \hat{D}_{z^2 - x^2 - y^2} = \frac{2}{3} \frac{\partial^2}{\partial z^2}-  \frac{1}{3} \frac{\partial^2}{\partial x^2}- \frac{1}{3} \frac{\partial^2}{\partial y^2} \]

\noindent
The remaining challenge is to calculate
the ground state wavefunction of the P impurity, which we do via a multi-million atom $sp^3d^5s^*$ atomistic tight-binding calculation
implemented with the \texttt{NEMO-3D} system \cite{tb_hamiltonian, nemo}. The tight-binding parameters have previously been fitted to accurately reproduce 
the donor energy spectrum \cite{central_cell} and the Si band structure \cite{tb_hamiltonian}, and is capable 
of accurately predicting important physical quantities including the Stark and strain-induced hyperfine shift \cite{hyperfine} and the anisotropic electron $g$ factor in 
strained Si \cite{gfactor}.  The calculations were performed over a $(40\mathrm{nm})^3$ domain consisting of around 3.1 million atoms, 
and include the effects of a central cell corrections \cite{central_cell}, the formation of Si dimer rows on the surface 
due to the 2$\cp$1 surface reconstruction of Si \cite{21_surface}, and hydrogen passivation \cite{h_passivation}. 
The agreement that has been demonstrated between experimental STM images and images calculated using this model is remarkable~\cite{usman2016spatial}.

\def\bibsection{\subsection*{\refname}}

\bibliographystyle{naturemag}
\bibliography{./refs}
\end{document}